\documentclass[acmtog,nonacm,balance=false]{acmart}
\usepackage{amsmath}
\usepackage{algorithm}
\usepackage{algpseudocode}
\usepackage{multirow}
\usepackage{placeins}
\setcopyright{none}
\newcommand{\wopfigure}[2]{\IfFileExists{#1}{\includegraphics[width=\linewidth,height=#2,keepaspectratio]{#1}}{\fbox{\parbox[c][#2][c]{0.94\linewidth}{\centering\small Figure asset not included\\\texttt{\detokenize{#1}}}}}}
\begin{document}
\title{Walk in Envelopes: Efficient Branching Walk-on-Spheres via Certified Solution Bounds}
\author{Yunfeng Liao}
\author{Zhenhao Qiu}
\author{Xinyang Chen}
\author{Xiucheng Li}
\affiliation{\institution{Harbin Institute of Technology (Shenzhen)}\city{Shenzhen}\state{Guangdong}\country{China}}
\renewcommand{\shortauthors}{Liao and Li}
\begin{abstract}
Branching Walk-on-Spheres (WoS) solves semilinear PDEs without meshes, but its computational cost grows sharply near critical branching rates. We introduce \emph{Walk in Envelopes} (WiE), which uses certified solution bounds to transform branching diffusion rather than prune existing trees. The transformation suppresses reproduction, introduces controlled termination, and preserves unbiasedness with an envelope-dependent variance bound. We realize the transformed process through exact rejection sampling using standard screened WoS kernels, supported by certified spline envelopes and adaptive sampling and  further prove that suitable harmonic envelopes yield bounded expected reproduction counts even as conventional BBM diverges. % Across four 2D domains, WiE achieves $6.9$--$16.2\times$ time--variance efficiency gains near criticality, excluding envelope preprocessing.
\end{abstract}
\begin{teaserfigure}
    \centering
    \includegraphics[width=1\textwidth]{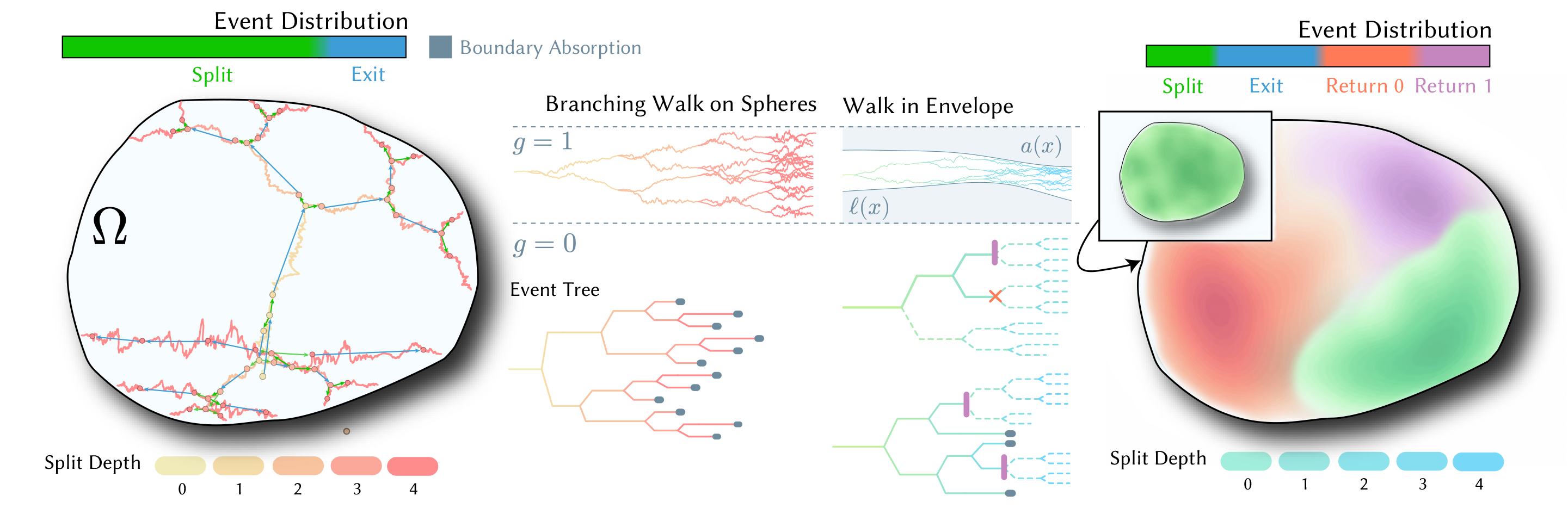}
    \caption{Walk in Envelopes (WiE) transforms branching WoS using certified lower and upper solution bounds. By reparameterizing the solution within a spatially varying envelope, WiE redistributes the original branching mechanism into splitting, return-0, and return-1 events, enabling early termination and pruning of redundant branches while preserving the target solution. The resulting computation trees are substantially smaller than those of conventional branching WoS.}
    \Description{1}
    \label{fig:teaser}
\end{teaserfigure}

\maketitle
\thispagestyle{fancy}\fancyfoot[C]
{\footnotesize Preprint}

\section{Introduction}
\label{sec:introduction}

Grid-free Monte Carlo methods have become an attractive alternative to mesh-based PDE solvers, enabling pointwise solution queries on complex domains without constructing a global discretization \cite{sawhney2020montecarlo}. While Walk-on-Spheres (WoS) efficiently resolves linear elliptic equations through boundary-hitting processes, extending this paradigm to semilinear PDEs requires handling nonlinear reactions. Branching Brownian motion (BBM) provides a natural probabilistic representation by replacing nonlinear terms with particle reproduction events \cite{agarwal2020branching}. Combined with WoS, this yields a grid-free solver for semilinear equations. However, eliminating Brownian time discretization does not eliminate branching complexity: each offspring initiates additional walks, and the expected particle population diverges as the branching intensity approaches a critical spectral threshold. This makes otherwise valid stochastic representations increasingly expensive precisely in strongly nonlinear regimes.

A natural strategy is to exploit approximate solution information to terminate branches whose remaining contributions are sufficiently small. Yet such pruning requires reliable bounds on unresolved subtrees, whose evaluation can itself be expensive. Approximate bounds may also introduce bias, while tighter estimates do not necessarily offset their computation costs. More fundamentally, pruning acts on a branching process that has already been generated, rather than preventing unnecessary reproduction. This motivates a different question: \textit{can certified solution information be used to construct an equivalent stochastic process that branches less frequently, without sacrificing correctness?}

We introduce \emph{Walk in Envelopes} (WiE), a branching WoS framework that transforms the stochastic representation using certified lower and upper solution envelopes. Given sub- and supersolutions $\ell\le u\le a$, we normalize the unknown solution as $v=(u-\ell)/(a-\ell)$. This induces a generalized Doob transform whose diffusion, reproduction, and termination mechanisms are determined by the envelope functions. For binary branching, the reproduction rate becomes $\lambda(a-\ell)$, while envelope residuals introduce additional success and failure termination events. The resulting estimator remains unbiased and has variance bounded by $(a-\ell)^2/4$. Moreover, we characterize its expected reproduction count through an elliptic equation, showing that a suitable harmonic envelope can keep this count uniformly bounded even as conventional BBM approaches its critical divergence threshold. Although branching-process transformations have established mathematical foundations \cite{englander2007branching}, WiE develops them into a certified computational framework for semilinear PDE solving.

The transformed process introduces spatially varying diffusion, which would ordinarily require simulating drifted trajectories. We avoid this difficulty by deriving an exact envelope-aware transition using the original screened WoS exit and Green kernels. Rejection sampling with certified local majorants realizes the transformed distribution without evaluating weighted kernel integrals. To obtain useful envelopes efficiently, we develop monotone refinement of cubic B-spline sub- and supersolutions, using Bernstein constraints to certify PDE inequalities throughout the domain by comparison principle. Adaptive screening, spatially localized rejection bounds, and lazy evaluation further reduce sampling overhead, while a spectral-workload criterion avoids unnecessary preprocessing in inexpensive regimes.

We evaluate WiE on four complex two-dimensional domains across branching intensities and accuracy requirements. Near criticality, WiE achieves $6.9$--$16.2\times$ improvements in time--variance efficiency and $4.7$--$14.6\times$ reductions in sampling time at matched accuracy, excluding one-time envelope construction. Its benefits increase with branching intensity, while a workload-aware fallback preserves the efficiency of conventional BBM in easier regimes. We additionally demonstrate extensions to projected walks on manifolds and semilinear parabolic equations.

Our main contributions are:
\begin{itemize}
    \item A certified envelope transformation of branching diffusion, with unbiased estimation, variance guarantees, and a spectral characterization of near-critical branching complexity.
    \item An exact envelope-aware WoS sampler using existing screened kernels, without explicitly simulating transformed diffusion.
    \item A practical solver combining certified spline-envelope refinement, efficient rejection sampling, and workload-aware fallback, with substantial near-critical acceleration.
\end{itemize}

\section{Related Work}
\label{sec:related_work}
\paragraph{Branching Diffusions and Solution-Guided Transformations.} BBM represents semilinear reaction--diffusion PDEs, beginning with McKean's Kolmogorov--Petrovskii--Piskunov representation \cite{mckean1975application}. Branching-diffusion solvers extend to polynomial nonlinearities, gradient dependence via marked particles \cite{henrylabordere2019branching}, and elliptic boundary-value problems on bounded domains \cite{agarwal2020branching}, but face integrability and branching-cost restrictions under strong branching. Process transformations are well established: nonlinear $h$-transforms alter both motion and reproduction \cite{englander2007branching}, while measure changes and backbone decompositions characterize conditioned branching and near-critical behavior \cite{salisbury2000nondegenerate,harris2016branching}. These underpin solution-guided representations. Our computational realization uses certified sub- and supersolutions to obtain nonnegative branching and killing rates, with exact sphere-transition sampling requiring neither the unknown solution nor drifted-trajectory simulation.

\paragraph{Monte Carlo Methods for Nonlinear Equations.} Monte Carlo solvers avoid global spatial discretization, but nested nonlinear expectations and recursive branching can sharply increase variance and work. Numerical branching methods handle polynomial nonlinear PDEs, subject to representation-dependent convergence and integrability \cite{henrylabordere2019branching,agarwal2020branching}. In graphics, \cite{tong2025practical} accelerate nonlinear recursive light transport through nonlinear path filtering that reuses intermediate evaluations and neural radiance caching that learns them. For nonlinear radiative boundary conditions, \cite{bao2026montecarlo} combine Picard-style iteration with Monte Carlo PDE solvers. These reduce nonlinear integral-equation costs through reuse, approximation, or iteration. Instead, we retain a branching estimator but change its probability law: certified envelopes suppress reproduction, bound its range, and preserve its expectation under stated assumptions. We reduce sample-generation work rather than only accelerate or approximate an existing recursion.

\paragraph{Walk-on-Spheres and Grid-Free PDE Solvers.} WoS originated for Dirichlet problems \cite{muller1956some} and entered geometric computing as a grid-free elliptic PDE solver for point queries \cite{sawhney2020montecarlo}. Extensions handle spatially varying coefficients via null scattering \cite{sawhney2022gridfree}, mixed Neumann--Dirichlet conditions via walk on stars \cite{sawhney2023walk}, and boundary-integral transport via walk on boundary \cite{sugimoto2023practical}. Boundary-value caching reuses estimates across queries \cite{miller2023boundary}, differential WoS computes parameter derivatives \cite{miller2024differential}, and projected walks solve surface PDEs \cite{sugimoto2024projected}. These advances mainly address linear PDE representations, boundary handling, or stochastic-estimate reuse. For semilinear PDEs, sphere jumps eliminate Brownian time discretization but not branching-tree growth. We target this remaining bottleneck through certified semilinear bounds and a transformed process with weighted exit and interior events sampled using the original screened WoS kernels.
\section{Envelope-Transformed Branching Walks}
\subsection{Preliminary: Branching Walk on Spheres}
Branching Brownian motion (BBM) provides a stochastic representation of semilinear elliptic PDEs by interpreting nonlinear reactions as particle branching events. Consider $F(u):=\mathcal L  u+\lambda(P(u)-u)=0$ with Dirichlet condition $u|_{\partial\Omega}=g$, where $P(u)=\sum_{k=0}^{K}p_ku^k$ is the offspring generating function and $\mathcal L = \frac12\Delta$. Each particle follows Brownian motion until either reaching the boundary or branching into $k$ independent offspring with probability $p_k$ at an exponential rate $\lambda$. The solution is recovered as $u(x)=\mathbb E[\prod_{i\in\mathcal T}g(X_i)]$, where $\mathcal T$ is the terminal leaves. 

For probabilistic semilinear elliptic PDEs, Dirichlet values naturally lie in $[0,1]$. For example, $g(y)\in[0,1]$ may specify the success probability of a particle absorbed at $y$, in which case $u(x)$ gives the probability that every terminal particle succeeds. This interpretation follows directly from the BBM representation $u(x)=\mathbb E[\prod_{i\in\mathcal L}g(X_i)]$, and ensures $0\le u\le1$.

Branching WoS accelerates this process by replacing Brownian trajectories with sphere transitions. At each particle position $x$, let $B=B(x,r)\subset\Omega$ be the largest inscribed ball, and denote by $H_\lambda$ and $G_\lambda$ its screened exit and Green kernels, respectively. Conditioning on whether the particle exits the ball before branching yields the integral equation induced by total expectation:
\begin{equation}
    u(x) = \int_{\partial B}H_\lambda(x,\mathrm{d}y)\,u(y)+\lambda\int_B G_\lambda(x,y)
\sum_{k=0}^{K}p_k u(y)^k\,\mathrm{d}y. 
\end{equation}

Specifically, the exit probability is $q_B=\int_{\partial B}H_\lambda(x,\mathrm{d}y)=1/I_0(r\sqrt{2\lambda})$ in $
 \mathbb R^2$, where $I_0$ denotes the modified Bessel function. With probability $q_B$, the particle jumps uniformly to $\partial B$ and continues walking; otherwise, an interior branching location $Y$ is sampled with density $\lambda G_\lambda(x,y)/(1-q_B)$, and $k\sim(p_k)$ offspring are added to the active particle queue at $Y$. The procedure repeats until all particles reach the domain boundary or terminate with zero offspring, producing one Monte Carlo sample from the product of terminal values. Thus, although WoS eliminates explicit trajectory discretization, each sphere transition still requires a branching decision, while every offspring introduces additional walks that must be maintained and evaluated. Canonical results show that BWoS crashes when the particle splitting rate surpasses the boundary-absorbing speed, i.e., $\lambda > \lambda_1$ where $\lambda_1$ is the least non-negative eigenvalue of $-\mathcal L$ in $\Omega$. 

\subsection{From Solution Bounds to Branching Processes}
Although WoS eliminates explicit Brownian trajectory discretization, its computational cost still grows with the branching tree. The expected cumulative particle population becomes unbounded as $\lambda$ approaches the principal Dirichlet eigenvalue $\lambda_1$ associated with $\Omega$, motivating the early termination of subtrees whose remaining contributions can be bounded. Let $\tau$ denote the first Brownian exit time from $\Omega$, and $h(x)=\mathbb E_x[g(B_\tau)]$ the harmonic extension of $g$. Several bounds follow directly from the BBM representation and its associated linear equations as Theorem~\ref{thm:prob-sol-bnds}. 

\begin{theorem}[Probabilistic Solution Bounds]
\label{thm:prob-sol-bnds}
For $0<\lambda<\lambda_1, P(u)=u^2$ and $0\le g\le1$, the solution of $F(u)=0$ with $u|{\partial\Omega}=g$ satisfies ${\max}\{L_0(x),L_1(x),L_2(x)\}\le u(x)\le U(x),$
where $U=h$, and $L_0(x)=\mathbb E_x[e^{-\lambda\tau}g(B\tau)]$, $L_1(x)=1-\mathbb E_x[e^{\lambda\tau}(1-g(B_\tau))]$, $L_2(x)=h(x)-\tfrac{\lambda}{4}\mathbb E_x[\tau]$. 
\end{theorem}
\begin{proof}
Since $F(u)=0$ implies $\mathcal Lu=\lambda u(1-u)\ge0$, the maximum principle gives $u\le h=U$. For $L_0$, restrict the BBM expectation to trajectories reaching the boundary before their first branching event; all omitted products are nonnegative, yielding $u\ge\mathbb E_x[e^{-\lambda\tau}g(B_\tau)]$. For $L_1$, let $q=1-u$. Then $\mathcal Lq+\lambda q=\lambda q^2\ge0$. The Feynman–Kac formula gives $q(x)=\mathbb E_x[e^{\lambda\tau}(1-g(B_\tau))]-\lambda\mathbb E_x[\int_0^\tau e^{\lambda t}q(B_t)^2\,\mathrm{d}t]$, hence $u\ge L_1$. Finally, Dynkin's formula gives $h-u=\lambda\mathbb E_x[\int_0^\tau u(B_t)(1-u(B_t))\,\mathrm{d}t]$. Since $u(1-u)\le1/4$, we obtain $u\ge L_2$.
\end{proof}

These bounds enable early termination of unresolved branching subtrees. Suppose the completed leaves contribute a product $\omega\in[0,1]$, while the active frontier contains particles at $x_1,\ldots,x_m$. By conditional independence, their remaining expected contribution is $R=\omega\prod_{i=1}^m u(x_i)$. With bounds $L\le u\le U$ immediately yield
\begin{equation}
B^-:=\omega\prod_{i=1}^m L(x_i)\le R\le
\omega\prod_{i=1}^m U(x_i):=B^+.
\label{eq:bbm-pruning}
\end{equation}

Thus, whenever $B^+-B^-\le2\varepsilon$, the remaining tree can be pruned and replaced by $(B^-+B^+)/2$, introducing at most $\varepsilon$ absolute bias in the final expectation, since the conditional expected contribution lies within $[B^-,B^+]$.

\begin{figure}[t]
    \centering
    \includegraphics[width=\linewidth]{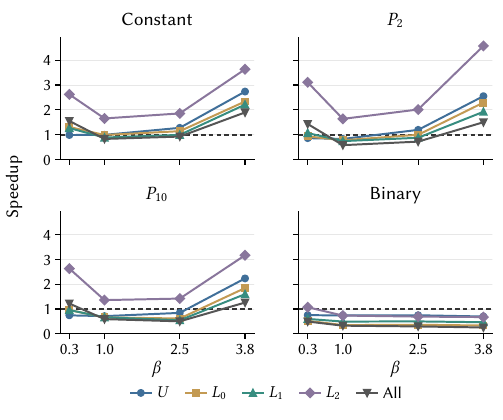}
    \caption{Pruning speedup relative to vanilla WoS--BBM (dashed $1\times$) across branching rates $\beta$, for constant, $P_2$, $P_{10}$, and binary-hemisphere boundary data (mean $0.5$).}
    \label{fig:wie_bbm_bounds_speedup_lines}
\end{figure}

Despite their validity, these bounds incur evaluation overhead and exhibit strongly case-dependent efficiency (Fig.~\ref{fig:wie_bbm_bounds_speedup_lines}). $L_0$ ignores branched trajectories, $L_1$ deteriorates with large boundary deficits, while $L_2$ relies on the worst-case bound $u(1-u)\le1/4$. Consequently, $L_2$ generally performs best on smooth boundaries, whereas discontinuous binary data defeats most pruning with binary boundary condition. Combining all bounds often increases overhead without proportional savings. This motivates certified envelopes that directly reduce branching rather than repeatedly pruning existing trees.

% why pruning is needed here since as branching rate approaches to $\lambda_1$, the particle amount never goes down. 

%Failure only use expensive and inaccurate bounds for pruning. hard to give a criterion working well in all scenes.  offer exp here

\subsection{Generalized Envelope Branching}
The preceding pruning strategy uses solution bounds to terminate subtrees after they have been generated, but does not alter the underlying branching dynamics. We instead construct an equivalent branching process whose reproduction and termination rates are determined directly by certified solution envelopes. This can be understood as a generalized Doob's $h$-transform, where the envelope width acts as the transforming function.

For $F(u)=0$, let $\ell$ and $a$ be certified sub- and supersolutions satisfying $0\le\ell\le u\le a\le1$, $F(\ell)\ge0$, and $F(a)\le0$. Define the domain envelop width $d:=a-\ell>0$, the normalized solution $v=(u-\ell)/d\in[0,1]$, and the nonnegative residuals $r_\ell=F(\ell)$ and $r_a=-F(a)$. By substituting $u=\ell+dv$ and expanding the offspring generating function, we obtain
$$
P(\ell+dv)=P(\ell)+\sum_{j=1}^K C_jv^j,
\qquad C_j=\frac{P^{(j)}(\ell)}{j!}d^j\ge0.
$$

The key is that the transformation changes not only the diffusion operator but also the branching mechanism into four parts -- drifted diffusion $\mathcal L^d v$, reduced branching $\sum_{j=2}^K {\lambda C_j (v^j-v)}/{d}$, returning one $r_c(1-v)/ d$ and zero (being killed) $-{r_a}v/{d}$ in the theorem below. 

\begin{theorem}[Generalized Envelope Branching]
With certified envelopes satisfying the conditions above and well-defined and integrable transformed branching process, $v=(u-\ell)/d$ satisfies
\begin{equation}
\begin{aligned}
0=\mathcal L^d v+\sum_{j=2}^K\frac{\lambda C_j}{d}(v^j-v) +  \frac{r_\ell}{d}(1-v)-\frac{r_a}{d}v,
\end{aligned}
\label{eq}
\end{equation}
where $\mathcal L^d=\mathcal L+\nabla\log d\cdot\nabla$. Consequently, the equivalent process branches into $j\ge2$ offspring at rate $\lambda C_j/d$, terminates successfully at rate $r_\ell/d$, and terminates unsuccessfully at rate $r_a/d$. If $W\in[0,1]$ is its terminal-product estimator, then $Z(x)=\ell(x)+d(x)W$ is unbiased $\mathbb E[Z(x)]=u(x)$ and 
\begin{equation}
\operatorname{Var}[Z(x)]
\le(u(x)-\ell(x))(a(x)-u(x))
\le{d(x)^2}/{4}.
\label{eq:var-gap}
\end{equation}
\end{theorem}

\begin{proof}
The product rule gives $\mathcal L(dv)=d\mathcal L^dv+v\mathcal Ld$, where $\mathcal L^dv=d^{-1}[\mathcal L(dv)-v\mathcal Ld]$. Since $F(a)-F(\ell)=-r_a-r_\ell$ and $P(a)-P(\ell)=\sum_{j=1}^K C_j$, we have $\mathcal Ld-\lambda d=-r_\ell-r_a-\lambda\sum_{j=1}^K C_j$. Substituting these identities into $F(\ell+dv)=0$ yields $d\mathcal L^dv+r_\ell(1-v)-r_av+\lambda\sum_{j=1}^KC_j(v^j-v)=0$. The $j=1$ term vanishes, giving Eq.~\eqref{eq}. Its nonnegative rates define the stated branching process, with boundary payoff $(g-\ell)/d\in[0,1]$. The branching representation gives $\mathbb E[W]=v$, hence $\mathbb E[Z]=u$. Finally, $Z\in[\ell,a]$ implies $\mathbb E[(Z-\ell)(a-Z)]\ge0$, which yields Eq.~\eqref{eq}.
\end{proof}

Beyond preserving the original PDE solution, it also changes the expected branching complexity. For conventional BBM, the expected number of reproduction events $m_B$ satisfies
$(-\mathcal L-\lambda)m_B=\lambda$, with homogeneous Dirichlet boundary conditions. The transformed process instead obeys the following equation.

\begin{lemma}[Expected Branching Complexity]
Let $m_{\ell,a}(x)$ denote the expected number of reproduction events in the transformed process, assumed finite. Then $H=m_{\ell,a}d$ satisfies $H|_{\partial\Omega}=0$ and 
\begin{equation}
\left[-\mathcal L+\lambda(1-P'(a))\right]H
=\lambda[P(a)-P(\ell)-P'(\ell)d].
\label{eq:branch-complexity}
\end{equation}
\end{lemma}

\begin{proof}
Writing $b_j=\lambda C_j/d$ and
$\kappa=(r_\ell+r_a)/d$, the first-event decomposition gives
\begin{equation}
-\mathcal L^d m_{\ell,a}
+\left[\kappa-\sum_{j=2}^{K}(j-1)b_j\right]m_{\ell,a}
=\sum_{j=2}^{K}b_j.
\end{equation}
Multiplying by $d$ and using
$\mathcal Ld=\lambda d-r_\ell-r_a-\lambda\sum_{j=1}^{K}C_j$,
together with
$\sum_{j=1}^{K}jC_j=P'(a)d$ and
$\sum_{j=2}^{K}C_j=P(a)-P(\ell)-P'(\ell)d$,
yields the stated equation.
\end{proof}

For binary branching, this reduces to $\left[-\mathcal L+\lambda(1-2a)\right]H
=\lambda d^2.$ Thus, the transformation changes both the source of branching work and the elliptic operator governing its propagation. In particular, narrowing the envelope suppresses the quadratic source term, while the upper envelope controls the effective reproduction potential. These effects explain how WiE can reduce branching complexity even when conventional BBM approaches its critical regime. Also, there exists a harmonic envelope such that the cost of WiE keeps finite as the splitting rate approaches the divergence limit of vanilla WoS in theory as Theorem~\ref{thm:bdd-near-critical} reveals.

\begin{theorem}[Bounded near-critical branching with harmonic envelopes]
\label{thm:bdd-near-critical}
    Let $\Omega\subset\mathbb R^n$ be a bounded, connected $C^{2,\alpha}$ domain, and let $\lambda_1>0$ be the principal Dirichlet eigenvalue of $-\mathcal L$, where $\mathcal L=\frac12\Delta$. Suppose $g\in C^{2,\alpha}(\partial\Omega)$ satisfies $0\le g\le1$, $g\not\equiv0$, and $g\not\equiv1$. Let $h$ be its harmonic extension, $\mathcal Lh=0$ in $\Omega$ and $h=g$ on $\partial\Omega$. For the binary branching equation $\mathcal Lu+\lambda(u^2-u)=0$, choose the certified envelopes $\ell=0$ and $a=h$. Then the expected reproduction count $m_\lambda$ of the transformed process satisfies
\begin{equation}
\sup_{0<\lambda<\lambda_1}|m_\lambda|_{L^\infty(\Omega)}<\infty.
\label{eq:wie-critical-bound}
\end{equation}
In contrast, the expected reproduction count $m_\lambda^{\mathrm{BBM}}$ of conventional binary BBM diverges as $\lambda\uparrow\lambda_1$: for every fixed $x\in\Omega$, $m_\lambda^{\mathrm{BBM}}(x)=\Theta((\lambda_1-\lambda)^{-1})$.
\end{theorem}
\begin{proof}
    The strong maximum principle gives $0<h<1$ in $\Omega$. Since $F(0)=0$ and $F(h)=\lambda(h^2-h)\le0$, the proposed envelopes are certified, with width $d=h>0$ in the interior. By Lemma 3.3, $H_\lambda:=hm_\lambda$ solves
\begin{equation}
A_\lambda H_\lambda=\lambda h^2,\qquad H_\lambda|{\partial\Omega}=0,\qquad
A\lambda=-\mathcal L-\lambda+2\lambda(1-h).
\label{eq:wie-critical-H}
\end{equation}
We first show that $A_\lambda$ has a spectral gap uniform up to $\lambda_1$. At $\lambda=\lambda_1$, its quadratic form is
\begin{equation}
\langle f,A_{\lambda_1}f\rangle
=\underbrace{\|\nabla f\|_2^2 /2-\lambda_1\|f\|_2^2}_{\ge0}
+2\lambda_1\int_\Omega(1-h)f^2\,\mathrm{d}x.
\label{eq:wie-critical-form}
\end{equation}
If its lowest eigenvalue were zero, equality would force $f$ to be a multiple of the positive first Dirichlet eigenfunction $\phi_1$ and simultaneously $\int_\Omega(1-h)f^2\,\mathrm{d}x=0$, contradicting $1-h>0$ in $\Omega$. Hence $\kappa:=\lambda_{\min}(A_{\lambda_1})>0$. Writing $t=\lambda/\lambda_1$ gives $A_\lambda=(1-t)(-\mathcal L)+tA_{\lambda_1}$, so the min–max principle yields
\begin{equation}
\lambda_{\min}(A_\lambda)\ge(1-t)\lambda_1+t\kappa\ge\min{\lambda_1,\kappa}>0.
\label{eq:wie-critical-gap}
\end{equation}
Uniform invertibility and standard Dirichlet elliptic estimates therefore imply $|H_\lambda|{C^1(\overline\Omega)}\le C$, with $C$ independent of $\lambda\in(0,\lambda_1)$. Since $H\lambda=0$ on $\partial\Omega$, $|H_\lambda(x)|\le C,\mathrm{dist}(x,\partial\Omega)$. On the other hand, positivity of $h$ in the interior, the Hopf boundary lemma at points where $g=0$, and compactness of $\partial\Omega$ imply $h(x)\ge c,\mathrm{dist}(x,\partial\Omega)$ for some $c>0$. Thus $m_\lambda=H_\lambda/h\le C/c$ throughout $\Omega$, proving the bound.
For conventional BBM, $(-\mathcal L-\lambda)m_\lambda^{\mathrm{BBM}}=\lambda$ with zero Dirichlet boundary data. The isolated principal-eigenvalue resolvent expansion gives
\begin{equation}
m_\lambda^{\mathrm{BBM}}(x)=
\frac{\lambda_1\langle 1,\phi_1\rangle\phi_1(x)}{\lambda_1-\lambda}+O_x(1),
\qquad \lambda\uparrow\lambda_1,
\label{eq:wie-critical-bbm}
\end{equation}
where $|\phi_1|_2=1$. This proves the stated divergence.
\end{proof}
% how an envelop squeezes down the branching rate by introducing a killing rate, which is in essence a sort of doob's h-transform. 

% Theorem 1 (Certified Envelope Branching Representation). variance upper bound for wie

% also, a theorem here to show that smaller gap does not mean lower branching rate. there exists a trade-off among low branching rate, low variance and rejection rate. 
\subsection{Walk in Envelopes: Exact Walk-on-Spheres Transition}

The generalized Doob transform introduces a spatially varying drift $\nabla\log d$, seemingly requiring a new sampler for the transformed diffusion, which instead be sampled exactly by reweighting the original WoS transitions, without explicitly simulating drifted trajectories.

% For a maximal inscribed ball $B=B(x,r)\subset\Omega$, let $H_\lambda(x,\mathrm{d}y)$ and $G_\lambda(x,y)$ denote the screened sphere-exit and Green kernels of $\lambda-\mathcal L$, respectively, satisfying $\int_{\partial B}H_\lambda+\lambda\int_B G_\lambda=1$. % Recall $d=a-\ell$, $r_\ell=F(\ell)$, $r_a=-F(a)$, and $C_j=P^{(j)}(\ell)d^j/j!$.

\begin{theorem}[Exact Envelope-Aware WoS Transition]
For any certified envelopes $0\le\ell<a\le1$ satisfying $F(\ell)\ge0$ and $F(a)\le0$, the normalized solution $v=(u-\ell)/d$ satisfies
\begin{equation}
\begin{aligned}
d(x)v(x)
={}&\int_{\partial B}H_\lambda(x,\mathrm{d}y)\,d(y)v(y)\\
&+\int_B G_\lambda(x,y)
\left[r_\ell(y)+\lambda\sum_{j=1}^{K}C_j(y)v(y)^j\right]\,\mathrm{d}y.
\end{aligned}
\label{eq:envelope-wos}
\end{equation}
Moreover, Eq.~\eqref{eq:envelope-wos}  admits an exact stochastic transition using only the original screened WoS kernels and pointwise envelope evaluations.
\end{theorem}

\begin{proof}
Expanding $F(\ell+dv)=0$ using $P(\ell+dv)=P(\ell)+\sum_{j=1}^KC_jv^j$ gives $(\lambda-\mathcal L)(dv)=r_\ell+\lambda\sum_{j=1}^KC_jv^j$. Applying the screened Green representation on $B$ establishes Eq.~\eqref{eq:envelope-wos}. Set $v\equiv1$ on the right-hand side of Eq.~\eqref{eq:envelope-wos} and have 
\begin{equation}
\lambda \omega_I
:=r_\ell+r_a+\lambda\sum_{j=1}^KC_j
=\lambda d-\mathcal Ld.
\label{eq:envelope-weight}
\end{equation}
Since all terms are nonnegative, $\omega_I\ge0$. Applying the screened Green representation to $d$ yields
\begin{equation}
    d(x)=\int_{\partial B}H_\lambda(x,\mathrm{d}y)d(y)
+\allowbreak \lambda \allowbreak  \int_B G_\lambda(x,y)\omega_I(y)\,\mathrm{d}y.
\label{eq:envelope-normalization}
\end{equation}
Dividing by $d(x)$ therefore defines a probability distribution over two mutually exclusive outcomes: exiting the sphere with density $H_\lambda(x,\mathrm{d}y)d(y)/d(x)$, or triggering an interior event with density $\lambda G_\lambda(x,y)\omega_I(y)/d(x)$. For an interior event at $y$, dividing the nonnegative terms in Eq.~\eqref{eq:envelope-weight} by $\lambda \omega_I(y)$ gives
\begin{equation}
\pi_+(y):=\frac{r_\ell(y)}{\lambda \omega_I(y)},\quad
\pi_-(y):=\frac{r_a(y)}{\lambda \omega_I(y)},\quad
\pi_j(y):=\frac{C_j(y)}{\omega_I(y)}.
\label{eq:envelope-event-probabilities}
\end{equation}
These probabilities correspond to successful termination, unsuccessful termination, and $j$ offspring, respectively, with $\pi_++\pi_-+\sum_j\pi_j=1$. Multiplying the interior density by each event probability recovers the corresponding integrand in Eq.~\eqref{eq:envelope-wos}; unsuccessful termination contributes zero. Hence the stochastic transition reproduces the original integral equation exactly.
\end{proof}

The remaining challenge is sampling this distribution without evaluating its weighted integrals. Although Eq.~\eqref{eq:envelope-normalization} provides the normalizing constant $d(x)$, directly sampling a density proportional to $H_\lambda d$ or $G_\lambda \omega_I$ would require a new sampler for each spatially varying envelope. Instead, we use rejection sampling to recover these densities from the original WoS proposals.

Specifically, the target exit and interior transition measures are proportional to $H_\lambda(x,\mathrm{d}y)d(y)$ and $\lambda G_\lambda(x,y)\omega_I(y)\,\mathrm{d}y$, respectively. We therefore introduce two certified constants $M_E\ge\sup_{\partial B}d$ and $M_I\ge\sup_B\omega_I$ to dominate these weighted measures by distributions that are already straightforward to sample. Let $q_B=\int_{\partial B}H_\lambda(x,\mathrm{d}y)$ denote the original sphere-exit probability. The corresponding proposal masses are $q_BM_E$ for exits and $(1-q_B)M_I$ for interior events, giving the computable normalization $Q_B=q_BM_E+(1-q_B)M_I$.

We first choose an exit proposal with probability $q_BM_E/Q_B$, otherwise an interior proposal. Its position $Y$ is drawn from $H_\lambda/q_B$ or $\lambda G_\lambda/(1-q_B)$, respectively, and accepted with probability $d(Y)/M_E$ or $\omega_I(Y)/M_I$. For example, the joint probability measure of proposing and accepting an exit at $y$ is $\frac{q_BM_E}{Q_B}
\frac{H_\lambda(x,\mathrm{d}y)}{q_B}
\frac{d(y)}{M_E}
=\frac{H_\lambda(x,\mathrm{d}y)d(y)}{Q_B}.$

Similarly, the accepted interior measure is $\lambda G_\lambda(x,y)\omega_I(y)\,\mathrm{d}y/Q_B$. Their total mass is $d(x)/Q_B$ by Eq.~\eqref{eq:envelope-normalization}. Conditioning on acceptance therefore divides both measures by $d(x)/Q_B$, recovering exactly the target transition probabilities. Thus, the accpetance rate is $d(x)/Q_B$ and rejected candidates can simply be resampled. Notice that neither the target normalization $d(x)$ nor any weighted Green integral needs to be evaluated during sampling.

% This construction uses the original Brownian sphere kernels regardless of whether the certified envelopes are harmonic. Their spatial variation is handled entirely by the acceptance weights and the residual-dependent interior events, while arbitrary offspring distributions enter only through the coefficients $C_j$.

\subsection{Extensions to Other Semilinear PDEs}
\label{sec:extensions}

Figure~\ref{fig:toy_manifold} compares projected WiE on a cylinder (top) and heat-ball WiE on the unit square (bottom) with independent finite-difference reference solutions.

\begin{figure}
    \centering
    \includegraphics[width=\linewidth]{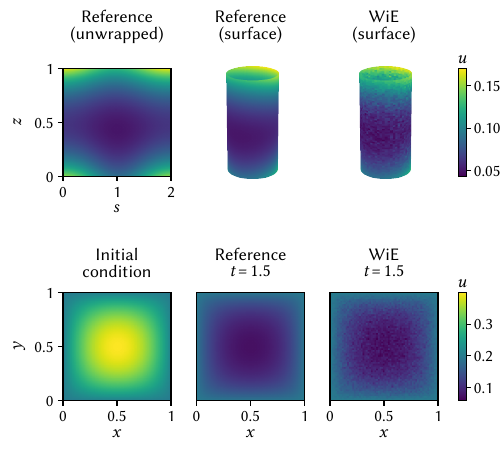}
    \caption{
Toy validation of two WiE extensions on the overall correctness and extensibility of WiE. Top: projected WiE on an embedded cylinder for $\frac12\Delta_{\mathcal M}u+\beta(u^2-u)=0$ with $\beta=0.97\lambda_1=0.97\pi^2/2$. Bottom: heat-ball WiE for $\partial_tu=\frac12\Delta u+\beta(u^2-u)$ on the unit square at $t=1.5$, with the same $\beta$. The six panels show, from left to right: cylinder reference in $(s,z)$ coordinates, cylinder reference on the embedded surface, projected WiE prediction, heat initial condition, heat reference at $t=1.5$, and heat-ball WiE prediction. We use 4096 full-field query points, with 72 samples per point on the cylinder and 128 samples per point for the heat equation. The measured time--variance speedups are about $24$--$54\times$ and $1.8$--$3.8\times$, respectively.
}
    \label{fig:toy_manifold}
\end{figure}

 \subsubsection{Projected Walk in Envelopes on Manifolds}
For an embedded manifold $\mathcal M\subset\mathbb R^n$, the envelope transformation extends by replacing $\mathcal L$ with the Laplace--Beltrami operator $\mathcal L_{\mathcal M}=\frac12\Delta_{\mathcal M}$, preserving the branching and termination coefficients. Following PWoS~\cite{Sugimoto:2024:PWoS}, we approximate intrinsic transitions using ambient Euclidean WoS kernels and the closest-point projection $y^*=\pi(y)$. Specifically, Eq.~\eqref{eq:envelope-wos} becomes

\begin{equation}
\begin{aligned}
v_\pi(x)&=\frac{1}{N_\pi(x)}\bigg[
\int_{\partial B}H_\lambda(x,\mathrm{d}y)\,d(y^*)v_\pi(y^*)\\
&+\int_B G_\lambda(x,y)
\left(r_\ell(y^*)+\lambda\sum_{j=1}^K C_j(y^*)v_\pi(y^*)^j\right)\,\mathrm{d}y\bigg],
\end{aligned}
\label{eq:projected-envelope-walk}
\end{equation}
where the projected normalization is
\begin{equation}
N_\pi(x)=\int_{\partial B}H_\lambda(x,\mathrm{d}y)d(y^*)
+\lambda\int_B G_\lambda(x,y)\omega_I(y^*)\,\mathrm{d}y.
\end{equation}

The rejection sampler from the preceding section applies unchanged, evaluating all envelope-dependent weights at $y^*$ instead of $y$. Since projected ambient kernels only approximate intrinsic Brownian transitions, $N_\pi(x)$ need not equal $d(x)$, introducing the same type of geometric approximation as PWoS without requiring intrinsic surface Green kernels.

\subsubsection{Semilinear Parabolic PDEs}
For $\partial_tu=\mathcal Lu+\lambda(P(u)-u)$, replace $\mathcal L$ by $\mathcal D=\mathcal L-\partial_t$ in the envelope residuals, requiring $r_\ell=\mathcal D\ell+\lambda(P(\ell)-\ell)\ge0$, $r_a=-\mathcal Da-\lambda(P(a)-a)\ge0$, and valid initial/lateral bounds. The coefficients $C_j$ and event rules retain their previous forms, but the interior weight becomes $W_\lambda=\lambda d-\mathcal Dd=r_\ell+r_a+\lambda\sum_j C_j$. Replace each spatial ball by a backward heat ball~\cite{Deaconu2017WOMS}, with elapsed time $s$ and displacement $z$: $\mathcal H_\eta=\{(s,z):0<s<\eta,\ |z|<R_\eta(s)\},$ $R_\eta(s)^2={ns\log(\eta/s)}$ with  $\eta<\min\{t,e\,\mathrm{dist}(x,\partial\Omega)^2/n\}.$ Its exit pair is sampled as $A\sim\Gamma(n/2+1,n/2),
S=\eta e^{-A},Y=R_\eta(S)\omega, \omega\sim\mathrm{Unif}(\mathbb S^{n-1}).$

Interior locations follow the \emph{killed occupation density} $q_\eta(s,z)\,\mathrm{d}s\,\mathrm{d}z=\mathbb P(B_s\in \mathrm{d}z,\,\sigma>s)\,\mathrm{d}s$, rather than the exit distribution: $q_\eta(s,z)=\big[p(s,z)-(2\pi\eta)^{-n/2}\big]\mathbf1_{(s,z)\in\mathcal H_\eta}$ with $p(s,z)=(2\pi s)^{-n/2}e^{-|z|^2/(2s)}$. Let $m_0=\int_{\mathcal H_\eta}q_\eta\,\mathrm{d}s\,\mathrm{d}z
=\eta\left(\frac n{n+2}\right)^{n/2+1}.$ To sample $q_\eta/m_0$, propose $S\sim\mathrm{Unif}(0,\eta)$ and $Y\sim\mathcal N(0,SI_n)$; accept with probability $\mathbf1_{(S,Y)\in\mathcal H_\eta}[1-(2\pi\eta)^{-n/2}/p(S,Y)]$. With $M_E\ge\sup_{\partial\mathcal H_\eta}d$ and $M_I\ge\sup_{\mathcal H_\eta}W_\lambda$, replace the preceding section's proposal masses by $M_E$ and $m_0M_I$ (thus $Q=M_E+m_0M_I$). At $(t-S,x+Y)$, accept exits with $e^{-\lambda S}d/M_E$ and interiors with $e^{-\lambda S}W_\lambda/M_I$; all offspring and termination rules remain unchanged. At the initial boundary, return $(u_0-\ell)/d$.

\section{Efficient Envelope Construction and Sampling}
\subsection{Monotone Envelope Refinement for Full Field}
% ways to build up good envelopes. 
% The preceding transformation allows any certified pair of envelopes to define an exact branching process. However, different envelopes can induce substantially different computational costs. The objective is therefore not merely to approximate the PDE solution, but to construct envelopes that minimize the cost of evaluating it over an entire domain.

The transformation admits many certified envelopes with different computational costs. We construct tight, reusable bounds for full-field evaluation through four steps: a comparison principle certifies bounds from PDE residuals and boundary conditions; query-weighted envelope width approximates the intractable time–variance objective; conservative affine constraints monotonically refine an initially feasible pair; and cubic B-splines with cellwise Bernstein constraints reduce continuous certification to finite-dimensional linear programs. Certification is preserved throughout, while preprocessing is amortized across repeated queries.

\begin{lemma}[Comparison Principle for Certified Envelopes]
For a bounded Lipschitz domain $\Omega\subset\mathbb R^n$, assume $\lambda\ge0$ and $\lambda\sup_{s\in[0,1]}\allowbreak (P'(s)-1)<\lambda_1.$ Suppose $u,\ell,a\in C^2(\Omega)\cap C(\overline\Omega)\cap H^1(\Omega)$ take values in $[0,1]$ and satisfy $F(u)=0, F(\ell)\ge0,  F(a)\le0$ and 
$\ell\le u\le a$ on $\partial\Omega$. Then
$\ell\le u\le a$ throughout $\Omega$.
\end{lemma}

\begin{proof}
We first prove $\ell\le u$. Set $w=\ell-u$. Since
$F(\ell)-F(u)\ge0$, subtracting the two PDE residuals gives $\mathcal Lw+\lambda\bigl(P(\ell)-P(u)-w\bigr)\ge0.$ To rewrite the nonlinear difference as a term proportional to $w$,
apply the fundamental theorem of calculus pointwise:
\begin{equation}
P(\ell(x))-P(u(x))
=\int_{u(x)}^{\ell(x)}P'(s)\,\mathrm{d}s
=w(x)\int_0^1P'\bigl(u(x)+t w(x)\bigr)\,\mathrm{d}t.
\end{equation}
Define the bounded coefficient $q(x)=\int_0^1P'\bigl(u(x)+t w(x)\bigr)\,\mathrm{d}t.$ Because $0\le u,\ell\le1$, every argument of $P'$ lies in $[0,1]$, and thus
\begin{equation}
\lambda(q(x)-1)
\le\lambda\sup_{s\in[0,1]}(P'(s)-1)<\lambda_1.
\end{equation}
It follows that $\frac12\Delta w+\lambda(q(x)-1)w\ge0,
\qquad w\le0\quad\text{on }\partial\Omega.$ We now justify the maximum-principle step rather than assuming it.
Let $w_+=\max\{w,0\}$. Since $w\le0$ on the boundary and
$w\in H^1(\Omega)$, we have $w_+\in H_0^1(\Omega)$. Multiplying the
last differential inequality by $w_+$ and integrating by parts gives
\begin{equation}
0\le-\frac12\int_\Omega|\nabla w_+|^2\,\mathrm{d}x
+\lambda\int_\Omega(q(x)-1)w_+^2\,\mathrm{d}x.
\end{equation}
Consequently,
\begin{equation}
\frac12\int_\Omega|\nabla w_+|^2\,\mathrm{d}x
\le c_*\int_\Omega w_+^2\,\mathrm{d}x,
\qquad c_*:=\lambda\sup_{s\in[0,1]}(P'(s)-1).
\end{equation}
The variational characterization of the first Dirichlet eigenvalue of
$-\frac12\Delta$ states that, for any $\varphi\in H_0^1(\Omega)$, $\frac12\int_\Omega|\nabla\varphi|^2\,\mathrm{d}x
\ge\lambda_1\int_\Omega\varphi^2\,\mathrm{d}x.$ Applying this inequality to $\varphi=w_+$ yields $(\lambda_1-c_*)\int_\Omega w_+^2\,\mathrm{d}x\le0.$
Since $\lambda_1-c_*>0$, we obtain $w_+=0$, hence $\ell\le u$.

To prove $u\le a$, repeat the same argument with $w=u-a$.
Indeed, $F(u)-F(a)\ge0$, $w\le0$ on the boundary, and $P(u)-P(a)=w\int_0^1P'\bigl(a+t(u-a)\bigr)\,\mathrm{d}t.$
The same upper bound on the resulting zero-order coefficient
and the same energy argument show $(u-a)_+=0$. Thus
$\ell\le u\le a$ in $\Omega$.
\end{proof}

\paragraph{Work-aware envelope optimization.}
Let $\mathcal Q\subset\Omega$ denote the evaluation region and let $T_{\ell,a}(x)$ and $V_{\ell,a}(x)$ be the expected runtime and estimator variance induced by a certified envelope pair $(\ell,a)$. For full-field evaluation, the ideal envelopes minimize the total cost of reaching a prescribed accuracy $\int_{\mathcal Q} T_{\ell,a}(x)V_{\ell,a}(x)\,\mathrm{d}\mu(x)$. Here $\mu$ describes the distribution of evaluation queries. The objective reflects the time--variance cost of independent Monte Carlo evaluation, but is difficult to optimize directly because $T_{\ell,a}$ depends on the entire transformed branching process.

Fortunately, the envelope width $d=a-\ell$ provides a tractable surrogate. From Eq.~\eqref{eq:envelope-wos}, the transformed reproduction coefficients are $C_j=P^{(j)}(\ell)d^j/j!$. In particular, for binary branching, the probability mass associated with reproduction is proportional to $d^2$. Moreover, Eq.~\eqref{eq:var-gap} bounds the estimator variance by $d^2/4$. These observations motivate minimizing the domain-averaged envelope width which replaces $T_{\ell ,a}(x)V_{\ell ,a}(x)$, yielding 
\begin{equation}
\begin{aligned}
\min_{\ell,a}\quad&
\int_{\mathcal Q} (a(x)-\ell(x))\,\mathrm{d}\mu(x),\\
\mathrm{s.t.}\quad&
F(\ell)\ge0,F(a)\le0,
0\le\ell\le a\le1,
\ell|_{\partial\Omega}\le g\le a|_{\partial\Omega}.
\end{aligned}
\label{eq:envelope-width}
\end{equation}

This objective seeks a single pair of reusable envelopes rather than independently optimizing bounds at each evaluation point. While width reduction does not guarantee a proportional runtime improvement, it simultaneously suppresses higher-order reproduction terms and tightens the estimator's variance bound.

\begin{figure}
    \centering
    \includegraphics[width=\linewidth]{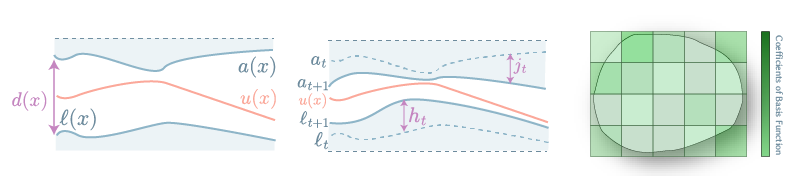}
    \caption{Starting from valid bounds $\ell(x)\leq u(x)\leq a(x)$, nonnegative corrections $h_t$ and $j_t$ monotonically tighten the lower and upper envelopes while preserving feasibility. Bernstein coefficients provide finite-dimensional sufficient conditions for continuous certification.}
    \label{fig:refinement}
\end{figure}

\paragraph{Monotone certified refinement.}
The optimization in Eq.~\eqref{eq:envelope-width} is nonlinear because the envelope constraints involve the reaction term $P$. We therefore begin with a feasible pair $(\ell_0,a_0)$ and construct a sequence of increasingly tight certified envelopes. At iteration $t$, we seek nonnegative corrections $h_t,j_t$ such that
\begin{equation}
\ell_{t+1}=\ell_t+h_t,\qquad
a_{t+1}=a_t-j_t.
\label{eq:envelope-refinement}
\end{equation}
The corresponding optimization maximizes the improvement
\begin{equation}
\begin{aligned}
\max_{h_t,j_t\ge0}\quad&
\int_{\mathcal Q}(h_t+j_t)\,\mathrm{d}\mu,\\
\mathrm{s.t.}\quad&
F(\ell_t+h_t)\ge0, \quad F(a_t-j_t)\le0,\quad \ell_t+h_t\le a_t-j_t,
\end{aligned}
\label{eq:envelope-refinement-opt}
\end{equation}
together with the original range and boundary constraints. Feasibility at every iteration implies
\begin{equation}
\ell_t\le\ell_{t+1}\le u\le a_{t+1}\le a_t,
\qquad
\mathcal J(\ell_{t+1},a_{t+1})
\le\mathcal J(\ell_t,a_t).
\label{eq:envelope-monotonicity}
\end{equation}
Thus, refinement never sacrifices a previously certified bound, and every accepted update improves the surrogate objective (Figure~\ref{fig:refinement}).

For a general offspring generating function $P(u)=\sum_{k=0}^{K}p_ku^k$, the nonlinear certification constraints can be converted into affine sufficient conditions by exploiting the convexity of $P$ on $[0,1]$. Substituting the monotone updates yields
\begin{equation}
\begin{aligned}
F(\ell_t+h_t)
&=F(\ell_t)+[\mathcal L+\lambda(P'(\ell_t)-1)]h_t
+\lambda R_{\ell}(h_t),\\
-F(a_t-j_t)
&=-F(a_t)+[\mathcal L+\lambda(P'(a_t)-1)]j_t
-\lambda R_a(j_t),
\end{aligned}
\end{equation}
where the Taylor remainders satisfy
\begin{equation}
\begin{aligned}
R_{\ell}(h_t)
&=\int_0^{h_t}(h_t-s)P''(\ell_t+s)\,\mathrm{d}s\ge0,\\
R_a(j_t)
&=\int_0^{j_t}(j_t-s)P''(a_t-s)\,\mathrm{d}s
\le\frac{M_t}{2}j_t^2,
\end{aligned}
\end{equation}
with $M_t(x)=\sup_{s\in[\ell_t(x),a_t(x)]}P''(s)$. Imposing a trust region $0\le j_t\le\delta$ further gives $R_a(j_t)\le M_t\delta j_t/2$. Therefore, the original nonlinear certification constraints are guaranteed by
\begin{equation}
\begin{aligned}
F(\ell_t)+[\mathcal L+\lambda(P'(\ell_t)-1)]h_t&\ge0,\\
-F(a_t)+[\mathcal L+\lambda(P'(a_t)-1-\tfrac12M_t\delta)]j_t&\ge0.
\end{aligned}
\label{eq:envelope-affine}
\end{equation}
Since $\ell_t$, $a_t$, and $M_t$ are fixed during each iteration, these conditions are affine in $h_t$ and $j_t$. Together with the linear width-reduction objective and monotonicity constraints, they yield a linear program after finite-dimensional representation and continuous certification. Thus, the refinement extends to general offspring distributions while preserving the validity and nesting of the certified envelopes.

% \paragraph{Finite-dimensional realization.}  We represent both envelopes using cubic B-splines, $\ell(x)=\sum_i c_i^\ell\phi_i(x)$ and $a(x)=\sum_i c_i^a\phi_i(x)$, where $\phi_i$ are locally supported basis functions and the spline coefficients are the optimization variables. At each refinement step, we optimize the coefficient increments to maximize the reduction in domain-averaged envelope width, subject to the affine PDE constraints, envelope ordering, and boundary conditions derived above. Since each B-spline is polynomial within a mesh cell, these constraints can be expressed in the local Bernstein basis, whose nonnegative coefficients guarantee that the corresponding inequalities hold everywhere inside that cell. Enforcing these coefficient constraints over all cells covering $\Omega$, together with the boundary constraints, therefore certifies the envelopes throughout the domain rather than only at discrete sample points. Since Bernstein coefficients depend linearly on the spline coefficients, the resulting optimization remains a linear program. The certified envelopes are constructed once and reused across full-field evaluations, amortizing preprocessing over all queries.

\paragraph{Finite-dimensional certified optimization.}
We represent the envelopes using cubic B-splines,
$\ell(x)=\sum_i c_i^\ell\phi_i(x)$ and
$a(x)=\sum_i c_i^a\phi_i(x)$, with their coefficients as optimization variables. Although Eq.~\eqref{eq:envelope-affine} is affine in the corrections, it imposes infinitely many spatial constraints. The following theorem provides a finite-dimensional certification. On each spline cell, we express the affine residuals and ordering constraints in the Bernstein basis, enforcing nonnegativity of their coefficients; range and polynomial boundary constraints are handled analogously. These finitely many sufficient linear inequalities certify the PDE constraints throughout the domain while preserving monotone envelope refinement.

For nonrectangular domains, we may restrict these certificates to a decomposition of each cell intersected with $\Omega$, using tensor-product or triangular Bernstein bases on the resulting pieces. This changes the sufficient finite constraints, while retaining the same globally $C^2$ spline representation and continuous PDE inequalities on $\Omega$. Therefore, the resulting refinement is a linear program with continuous, rather than pointwise, certification. The Bernstein constraints yields a conservative feasible set. The envelopes are constructed once and reused across full-field evaluations, amortizing their preprocessing cost over all queries.

\subsection{Efficient Envelope-Aware Sampling}

Although the envelope transformation reduces branching complexity, sampling its spatially varying transition weights introduces additional overhead. We employ three complementary optimizations to reduce this cost while preserving the exact transition law.

\paragraph{Adaptive screening.} The screened-kernel representation allows an adjustable parameter $\mu\in[0,\lambda]$. Using Eq.~\eqref{eq}, the interior weight $W_\mu =\lambda \omega_I-(\lambda-\mu)d$, yielding 
\begin{equation}
\begin{aligned}
W_\mu =r_\ell+r_a+\lambda\sum_{j=2}^{K}C_j
+[\mu-\lambda+\lambda P'(\ell)]d.
\end{aligned}
\end{equation}
This decomposition shows that changing $\mu$ affects only the single-offspring continuation weight, leaving termination and genuine reproduction weights unchanged. Their nonnegativity is preserved by choosing $\mu_B=\lambda\max\{0,1-P'(\underline\ell_B)\}$, where $\underline\ell_B$ is a certified lower bound over the sphere. A smaller $\mu_B$ increases the sphere-exit probability, potentially reducing expensive interior proposals without altering the represented solution.

\paragraph{Spatially adaptive majorants.}
Instead of using global rejection bounds, we independently bound the exit and interior weights by $M_E\ge\sup_{\partial B}d$ and $M_I\ge\sup_B W_\mu$. Their proposal masses are $q_\mu M_E$ and $m_\mu M_I$, where $q_\mu=\int_{\partial B}H_\mu$ and $m_\mu=\int_B G_\mu$. The corresponding acceptance probabilities are $d(y)/M_E$ and $W_\mu(y)/M_I$. We obtain tight local bounds from the spline cells intersecting each sphere and accelerate their aggregation using range-minimum/maximum queries (RMQ), reducing rejection without changing the accepted distribution.

\paragraph{Lazy evaluation.}
Even with tighter proposals, computing the exact envelope weights for every candidate is unnecessary. We first compare the rejection threshold against inexpensive local lower and upper bounds: candidates are accepted or rejected immediately when these bounds determine the outcome, and full spline evaluation is performed only otherwise. Green-kernel sampling is deferred until an interior proposal is selected, while reusable kernel quantities are cached. Thus, most inexpensive decisions avoid the cost of exact weight evaluation, and the savings accumulate across repeated transitions in the branching tree.

\begin{figure}
    \centering
    \includegraphics[width=\linewidth]{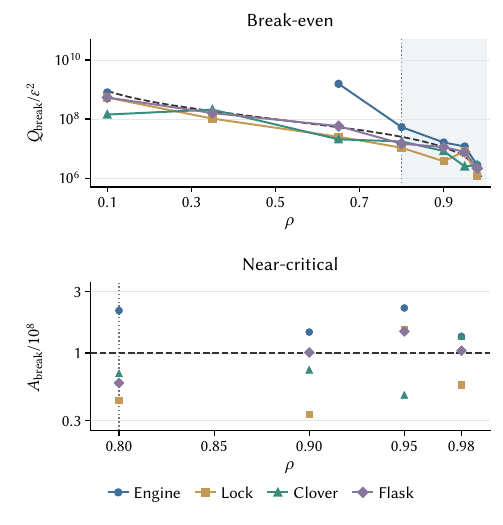}
    \caption{Empirical validation of the pilot-free spectral-workload fallback on four benchmark scenes. Top: the break-even workload $Q_{\mathrm{break}}/\varepsilon^2$ at $\varepsilon=0.003$ compared with the proposed boundary $A(1-\rho)/\rho$ (dark dashed curve), where $A=10^8$. At $\rho=0.35$ in Engine, WiE has higher equal-accuracy sampling cost than BBM, so no finite query count can amortize its preprocessing under the measured cost model. Bottom: unconnected points show the normalized scale $A_{\mathrm{break}}/A$, where $A_{\mathrm{break}}=(Q_{\mathrm{break}}/\varepsilon^2)\rho/(1-\rho)$, across five target accuracies. The horizontal dashed line marks $A_{\mathrm{break}}/A=1$. Their near-critical concentration around one motivates the empirical rule $\rho\ge0.8$ and $\mathcal W\ge10^8$. Vertical dotted lines mark $\rho=0.8$; the shaded region in the top panel indicates $\rho\ge0.8$. Costs combine recorded preprocessing with equal-RMSE sampling estimates.}
    \label{fig:fallback}
\end{figure}

\paragraph{Spectral-workload fallback.} Envelope construction incurs a one-time cost that may outweigh its sampling benefits. Motivated by the near-critical BBM work scaling $(1-\rho)^{-1}$, where $\rho=\lambda(P'(1)-1)/\lambda_1<1$, we define a spectral-workload indicator for $Q$ queries at target RMSE $\varepsilon$:
\begin{equation}
\mathcal W=\frac{Q}{\varepsilon^2}\frac{\rho}{1-\rho}.
\label{eq:spectral-workload}
\end{equation}
We activate WiE only when $\rho\ge\rho_0$ and $\mathcal W\ge A$, with $\rho_0=0.8$ and $A=10^8$ calibrated empirically across four scenes (Figure~\ref{fig:fallback}). This pilot-free criterion avoids unnecessary preprocessing, while imperfect fallback decisions affect efficiency but not correctness.

\section{Experiments}
\label{sec:experiments}
Our proposed WiE algorithm is evaluated in four non-trivial 2D scenes against canonical branching WoS. The results are illustrated in Figure~\ref{fig:2D_test} and the its speedup effects are evaluated with different branching rates. In Table~\ref{tab:wie-branch-rate}, it is clear that WiE naturally appears insignificant with low splitting rate and thus the extra processing cost overwhelms its variance advantage per sample; while its merit become explicit as $\beta /\lambda_1$ grows, it strongly suppresses the splitting rate in simulation by being enveloped, which is consistent across all scenarios. Besides, WiE also achieves the same accuracy using much less time than vanilla branching WoS in high splitting scenes as in Table~\ref{tab:wie-rmse-scaling}. 

\begin{figure}[tbp]
\centering
\includegraphics[width=\columnwidth]{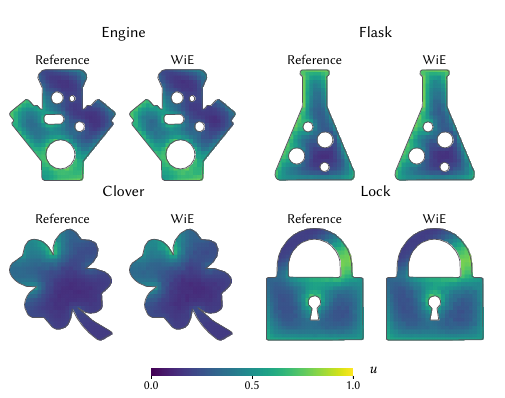}
\caption{Two-dimensional fields in four domains. Reference and WiE use 1,000 and 600 samples per query, respectively. The reference is itself a Monte Carlo estimate; these are not equal-time results. Geometry and queries share the same rigid $180^\circ$ rotation, and all fields use the same color range. Nearest-query display preserves the original values.}
\Description{Upright Engine, Flask, Clover, and Lock fields, each with reference and WiE side by side.}
\label{fig:2D_test}
\end{figure}

\begin{table}[t]
\centering
\caption{Efficiency across branching intensities. Single-sample variances are averaged over query points; the BBM and WiE work columns report $10^3\sum_i t_i\operatorname{Var}[Z_i]$ (lower is better), where $t_i$ is the CPU time measured at 256 samples per query. Speedup is the ratio of these work measures. All results average three seeds; envelope preprocessing is excluded.}
\label{tab:wie-branch-rate}
\scriptsize
\setlength{\tabcolsep}{2.8pt}
\renewcommand{\arraystretch}{1.08}
\resizebox{\columnwidth}{!}{%
\begin{tabular}{lc cc cc c}
\toprule
 & & \multicolumn{2}{c}{Single-sample variance} & \multicolumn{2}{c}{$10^3$ time$\times$variance} & \\
\cmidrule(lr){3-4}\cmidrule(lr){5-6}
Scene & $\beta/\widehat\lambda_1$ & BBM & WiE & BBM & WiE & Speedup \\
\midrule
\multirow{5}{*}{Engine} & 0.10 & 0.0066 & 0.0042 & 1.48 & 1.25 & 1.18$\times$ \\
 & 0.85 & 0.0262 & 0.0322 & 28.60 & 12.06 & 2.37$\times$ \\
 & 0.90 & 0.0271 & 0.0332 & 41.26 & 12.89 & 3.20$\times$ \\
 & 0.95 & 0.0283 & 0.0342 & 71.29 & 13.00 & 5.48$\times$ \\
 & 0.98 & 0.0287 & 0.0348 & 113.36 & 13.40 & \textbf{8.46$\times$} \\
\midrule
\multirow{5}{*}{Lock} & 0.10 & 0.0101 & 0.0068 & 3.02 & 2.82 & 1.07$\times$ \\
 & 0.85 & 0.0362 & 0.0337 & 71.68 & 16.67 & 4.30$\times$ \\
 & 0.90 & 0.0373 & 0.0350 & 102.66 & 18.02 & 5.70$\times$ \\
 & 0.95 & 0.0383 & 0.0358 & 187.15 & 19.12 & 9.79$\times$ \\
 & 0.98 & 0.0389 & 0.0360 & 294.56 & 18.96 & \textbf{15.54$\times$} \\
\midrule
\multirow{5}{*}{Clover} & 0.10 & 0.0158 & 0.0065 & 5.35 & 2.72 & 1.97$\times$ \\
 & 0.85 & 0.0257 & 0.0205 & 51.36 & 8.87 & 5.79$\times$ \\
 & 0.90 & 0.0260 & 0.0211 & 73.38 & 9.58 & 7.66$\times$ \\
 & 0.95 & 0.0262 & 0.0218 & 118.56 & 9.27 & 12.79$\times$ \\
 & 0.98 & 0.0264 & 0.0222 & 156.26 & 9.66 & \textbf{16.18$\times$} \\
\midrule
\multirow{5}{*}{Flask} & 0.10 & 0.0157 & 0.0106 & 2.78 & 2.47 & 1.13$\times$ \\
 & 0.85 & 0.0379 & 0.0429 & 27.80 & 12.39 & 2.24$\times$ \\
 & 0.90 & 0.0391 & 0.0447 & 37.54 & 12.93 & 2.90$\times$ \\
 & 0.95 & 0.0402 & 0.0405 & 63.66 & 12.15 & 5.24$\times$ \\
 & 0.98 & 0.0409 & 0.0431 & 87.08 & 12.58 & \textbf{6.92$\times$} \\
\bottomrule
\end{tabular}}
\end{table}

\begin{table}[t]
\centering
\caption{Equal-accuracy cost under a fixed near-critical branching intensity $\rho=\beta/\widehat\lambda_1=0.99$. We vary the target stochastic RMSE $\varepsilon$ over one order of magnitude; $N$ is the number of independent trees per query. CPU time (seconds) is summed over 48 fixed queries and averaged across three seeds, excluding envelope preprocessing. Sample budgets are variance-calibrated, and speedup is the ratio of measured CPU costs.}
\label{tab:wie-rmse-scaling}
\scriptsize
\setlength{\tabcolsep}{3pt}
\renewcommand{\arraystretch}{1.10}
\resizebox{\columnwidth}{!}{%
\begin{tabular}{lc rr rr c}
\toprule
& Target & \multicolumn{2}{c}{Samples/query} & \multicolumn{2}{c}{CPU (s)} & \\
\cmidrule(lr){3-4}\cmidrule(lr){5-6}
Scene & RMSE & BBM & WiE & BBM & WiE & Speedup \\
\midrule
\multirow{4}{*}{Engine} & 0.020 & 76 & 101 & 0.099 & 0.021 & 4.74$\times$ \\
 & 0.010 & 304 & 402 & 0.414 & 0.088 & 4.69$\times$ \\
 & 0.005 & 1,216 & 1,608 & 1.624 & 0.324 & 5.01$\times$ \\
 & 0.002 & 7,600 & 10,050 & 9.901 & 1.977 & 5.01$\times$ \\
\midrule
\multirow{4}{*}{Lock} & 0.020 & 96 & 85 & 0.248 & 0.020 & 12.44$\times$ \\
 & 0.010 & 382 & 337 & 0.992 & 0.077 & 12.91$\times$ \\
 & 0.005 & 1,528 & 1,348 & 4.185 & 0.298 & 14.04$\times$ \\
 & 0.002 & 9,550 & 8,425 & 24.628 & 1.861 & 13.23$\times$ \\
\midrule
\multirow{4}{*}{Clover} & 0.020 & 66 & 50 & 0.134 & 0.010 & 13.16$\times$ \\
 & 0.010 & 261 & 200 & 0.532 & 0.043 & 12.26$\times$ \\
 & 0.005 & 1,044 & 800 & 2.129 & 0.146 & 14.61$\times$ \\
 & 0.002 & 6,525 & 5,000 & 13.151 & 0.961 & 13.69$\times$ \\
\midrule
\multirow{4}{*}{Flask} & 0.020 & 107 & 112 & 0.143 & 0.026 & 5.50$\times$ \\
 & 0.010 & 427 & 447 & 0.573 & 0.096 & 5.97$\times$ \\
 & 0.005 & 1,708 & 1,788 & 2.133 & 0.382 & 5.58$\times$ \\
 & 0.002 & 10,675 & 11,175 & 13.314 & 2.399 & 5.55$\times$ \\
\bottomrule
\end{tabular}}
\end{table}

\paragraph{Envelope construction ablation.}
We compare three certified constructions with the same $6\times6$ bicubic basis: \emph{sequential LP} optimizes the upper and then lower envelope separately, \emph{joint LP} minimizes their width together, and \emph{monotone refinement} further tightens the joint solution by certified nonlinear iterations. The archived three-update configuration in Figure~\ref{fig:construction-ablation} shows that joint optimization improves width over sequential fitting, while refinement additionally reduces it by $5.5\%$ on Engine and $7.2\%$ on Flask without increasing spline resolution. Gains on Lock and Clover are negligible, and tighter envelopes do not uniformly imply lower sampling work. Thus, iteration can improve a fixed representation, but is not a universal substitute for increasing its resolution.
\begin{figure}[t]
    \centering
    \includegraphics[width=\linewidth]{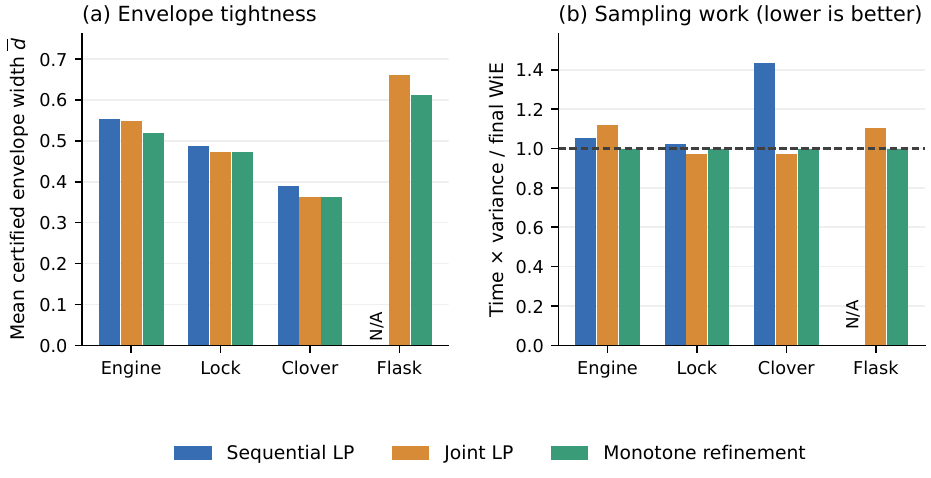}
    \caption{Construction ablation at $\rho=0.98$: (a) mean certified envelope width and (b) sampling time--variance cost relative to refined WiE. All methods use the same $6\times6$ spline basis. The unavailable Flask sequential certificate is marked N/A.}
    \label{fig:construction-ablation}
\end{figure}

\paragraph{Spline resolution.}
Table~\ref{tab:wie-spline-resolution-iterations}(a) examines spatial resolution and geometric certification. The earlier implementation certifies the entire bounding square, including regions outside the PDE domain and inside holes. We instead clip certificate cells to the original polygon: intact or axis-aligned pieces use tensor-product Bernstein coefficients, and oblique pieces use triangular Bernstein coefficients. The envelopes remain globally $C^2$ cubic B-splines. We keep the same boundary data, offspring laws, and $\rho=\lambda(P'(1)-1)/\widehat\lambda_1=0.98$. For all six non-square domains, we transfer the same $k=6$ joint initialization exactly to $k=12$ and $24$ before refinement, so the resolution comparison does not change the initial functions. Both coefficient vectors are restricted to $[-2,2]$ to suppress large, weakly constrained exterior modes. This extension constraint is additional to the certified range and PDE conditions inside $\Omega$.

\paragraph{Refinement iterations.}
We allow up to 12 updates, retaining $d\ge0.003$ and the upper-correction limit $0.075$. We stop after two successive width reductions below $10^{-6}$, or return the last valid pair when the LP budget or numerical checks fail; the actual accepted count is reported. If the full polynomial LP accepts no update, we retry from the same initialization using conservative cellwise scalar bounds on the residuals and $P'$; both attempts are included in construction cost. Table~\ref{tab:wie-spline-refinement-detail}(b) tests the returned pairs with fixed queries, sample count, seeds, and rejection-bound grid. At $k=24$, the six non-square cases reduce mean width by 53--87\% relative to the earlier coarse result. Relative to their common initialization, their sampling time--variance efficiency improves by $2.06$--$23.13\times$, while additional construction costs $298$--$600$ CPU seconds. Thus finer envelopes improve sampling in these tests, but need sufficient reuse to offset preprocessing. The L-shaped case attains a smaller mean width at $k=12$ than at $k=24$; the latter uses the conservative fallback, and these finite-budget results do not imply monotonic improvement with resolution. These results diagnose envelope construction rather than changing the variable transformation underlying WiE. All cases use the general-PGF kernel; Figure~\ref{fig:construction-ablation} retains the earlier specialized binary implementation.

\begin{table}[t]
\centering
\caption{Near-critical spline construction at $\rho=0.98$. (a) Mean width $\bar d$ on the fixed construction queries. \emph{Box} is the earlier full-box $k=6$ refinement; the other columns use certificates restricted to the geometric domain. All resolutions start from exactly the same $k=6$ joint envelope, transferred by knot insertion, with $(k+3)^2$ coefficients per side. We sweep $k=6,12,24$ on all six non-square domains; the three square cases remain $k=6$ controls. Refinement allows at most 12 updates, with a two-update width-plateau threshold of $10^{-6}$. A dagger indicates termination by the LP budget or numerical failure, rather than convergence. $j_*$ counts accepted updates at the resolution used in (b). At $k=24$, the scalar fallback is used for L-shaped, Two holes.}
\label{tab:wie-spline-resolution-iterations}
\small
\setlength{\tabcolsep}{3pt}
\renewcommand{\arraystretch}{1.05}
\begin{tabular}{lrrrrr}
\toprule
\multicolumn{6}{l}{(a) Resolution and geometric certification}\\
Case & Box & $k=6$ & $k=12$ & $k=24$ & $j_*$\\
\midrule

Engine & 0.5180 & 0.4095 & 0.2212$^\dagger$ & 0.1452$^\dagger$ & 2\\
Lock & 0.4574 & 0.3167$^\dagger$ & 0.1716$^\dagger$ & 0.0916$^\dagger$ & 3\\
Clover & 0.3621 & 0.2803 & 0.2019$^\dagger$ & 0.1543$^\dagger$ & 2\\
Flask & 0.6124 & 0.4103 & 0.1818$^\dagger$ & 0.0777$^\dagger$ & 6\\
L-shaped & 0.0517 & 0.0421 & 0.0081 & 0.0104 & 7\\
Two holes & 0.3819 & 0.3712$^\dagger$ & 0.2717 & 0.1795$^\dagger$ & 4\\
Square ($u^2$) & 0.0075 & 0.0075 & -- & -- & 9\\
Square ($u^4$) & 0.0062 & 0.0062$^\dagger$ & -- & -- & 4\\
Square ($p_0>0$) & 0.0057 & 0.0057$^\dagger$ & -- & -- & 4\\
\bottomrule
\end{tabular}
\end{table}

\begin{table}[t]
\ContinuedFloat
\centering
\caption{(b) Sampling with the common initialization ($j=0$) and the returned envelope. Every configuration uses the same 48 queries per case, $N=1024$ samples per query, five seeds, one general-PGF kernel, and a fixed $96\times96$ rejection-bound grid. $t_r$ is additional construction CPU time beyond the shared initialization, including unsuccessful LP attempts and any certified scalar-bound fallback; $t_s$ is sampling CPU time, both in seconds. $E$ is reference RMSE. $S_0=C_0/C$, with $C=\sum_q(t_q/N)\widehat V_q$, measures sampling time--variance efficiency and excludes construction. It is not an end-to-end speedup. References are independent 65,536-sample BBM estimates for the first four cases and an $801^2$ finite-difference solution, checked at $401^2$, for the remaining cases.}
\label{tab:wie-spline-refinement-detail}
\small
\setlength{\tabcolsep}{2.3pt}
\renewcommand{\arraystretch}{1.02}
\begin{tabular}{lrrrrrr}
\toprule
\multicolumn{7}{l}{(b) Initialization and refined sampling}\\
Case & $k$ & $j$ & $t_r$ & $t_s$ & $10^3E$ & $S_0$\\
\midrule

Engine & 6 & 0 & 0 & 0.614 & 5.745 & 1.00\\
 & 24 & 2$^\dagger$ & 365.9 & 0.635 & 2.204 & 9.41\\
\addlinespace[2pt]
Lock & 6 & 0 & 0 & 0.507 & 5.868 & 1.00\\
 & 24 & 3$^\dagger$ & 383.6 & 0.870 & 1.785 & 11.60\\
\addlinespace[2pt]
Clover & 6 & 0 & 0 & 0.484 & 4.496 & 1.00\\
 & 24 & 2$^\dagger$ & 324.4 & 0.931 & 2.650 & 2.06\\
\addlinespace[2pt]
Flask & 6 & 0 & 0 & 0.439 & 7.541 & 1.00\\
 & 24 & 6$^\dagger$ & 600.1 & 0.683 & 1.462 & 23.13\\
\addlinespace[2pt]
L-shaped & 6 & 0 & 0 & 0.539 & 0.775 & 1.00\\
 & 24 & 7 & 332.6 & 0.901 & 0.140 & 13.18\\
\addlinespace[2pt]
Two holes & 6 & 0 & 0 & 1.051 & 5.633 & 1.00\\
 & 24 & 4$^\dagger$ & 298.0 & 0.994 & 2.815 & 6.35\\
\addlinespace[2pt]
Square ($u^2$) & 6 & 0 & 0 & 0.649 & 0.783 & 1.00\\
 & 6 & 9 & 2.6 & 0.777 & 0.116 & 41.56\\
\addlinespace[2pt]
Square ($u^4$) & 6 & 0 & 0 & 0.685 & 0.450 & 1.00\\
 & 6 & 4$^\dagger$ & 5.9 & 0.812 & 0.094 & 22.58\\
\addlinespace[2pt]
Square ($p_0>0$) & 6 & 0 & 0 & 0.649 & 0.423 & 1.00\\
 & 6 & 4$^\dagger$ & 3.9 & 0.830 & 0.089 & 24.62\\
\addlinespace[2pt]
\bottomrule
\end{tabular}
\end{table}

\paragraph{Sampling optimization ablation.}
Figure~\ref{fig:sampling-ablation} isolates Bernstein squeeze and RMQ using the same certified envelopes and random paths ($\rho=0.98$). Squeeze reduces full spline evaluations from $14$--$17$ to $2.5$--$4.7$ per root sample; disabling it increases CPU time by $6.9$--$15.6\%$. Disabling both optimizations increases CPU by $9$--$16\%$. RMQ alone offers limited and sometimes negative gains for these small auxiliary grids. Crucially, these optimizations leave the sampled estimator unchanged; the improvements arise from cheaper decisions rather than a different probability law.
\begin{figure}[t]
    \centering
    \includegraphics[width=\linewidth]{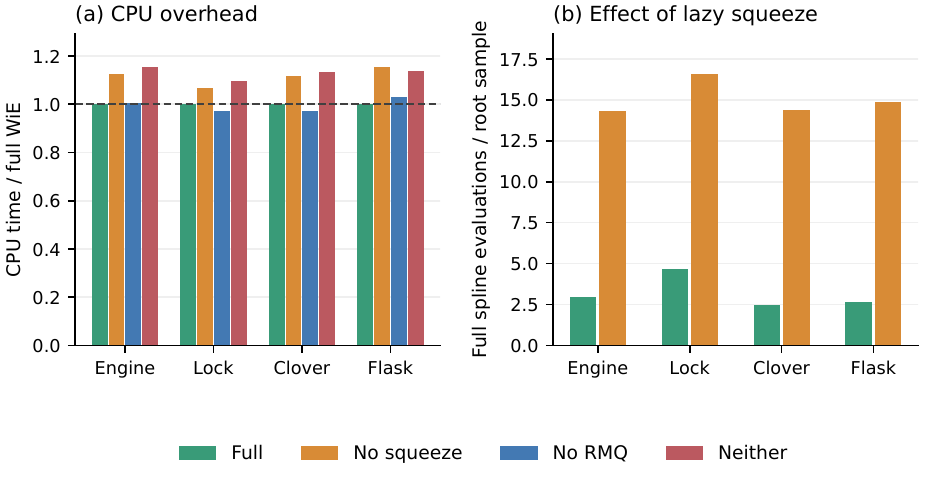}
    \caption{Engineering ablation on four scenes at $\rho=0.98$: (a) CPU time relative to full WiE, and (b) full spline evaluations per root sample with and without lazy squeeze. Each configuration uses identical seeds and $512$ samples per query.}
    \label{fig:sampling-ablation}
\end{figure}

\paragraph{General offspring distributions.}
To validate the polynomial PGF formulation beyond binary branching, we evaluate an independent general-offspring WoS implementation with certified local $3\times3$ B-spline envelopes on five two-dimensional problems, including pure quartic, mixed, and zero-offspring reactions (Table~\ref{tab:general-pgf}). WiE improves time--variance efficiency across these cases, although the benefit depends strongly on geometry, boundary values, and the attainable envelope width. The table provides a generality test rather than a controlled comparison of offspring degree: reaction rates and domains vary between cases.
\begin{table}[t]
\centering
\caption{General PGF validation with $700$ samples per query. RMSE is measured against independent reference fields. Speedup is the BBM/WiE time--variance cost ratio (excluding local envelope construction). The local-spline experiment uses a separate general-PGF solver without multi-round refinement.}
\label{tab:general-pgf}
\scriptsize
\setlength{\tabcolsep}{3pt}
\renewcommand{\arraystretch}{1.08}
\resizebox{\columnwidth}{!}{%
\begin{tabular}{llrrr}
\toprule
Case & $P(u)$ & BBM RMSE & WiE RMSE & Speedup \\
\midrule
Square (binary)  & $u^2$                                  & 0.00355 & 0.00107 & 29.91$\times$ \\
Square (quartic) & $u^4$                                  & 0.00340 & 0.00090 & 60.04$\times$ \\
L-shaped         & $0.25u^2+0.45u^3+0.30u^5$             & 0.00287 & 0.00139 &  5.43$\times$ \\
Two holes        & $0.10u+0.45u^2+0.45u^4$               & 0.00969 & 0.01039 &  1.27$\times$ \\
Square ($p_0>0$) & $0.10+0.55u^2+0.35u^3$                & 0.00496 & 0.00338 &  2.39$\times$ \\
\bottomrule
\end{tabular}}
\end{table}

\paragraph{Full-tree criticality in a ball.}
Figure~\ref{fig:wiefig-criticality} illustrates the first-moment change for binary branching in the three-dimensional unit ball, with $g\equiv0.75$ and the constant envelope $[0,0.75]$.

\begin{figure}[tbp]
\centering
\includegraphics[width=\columnwidth]{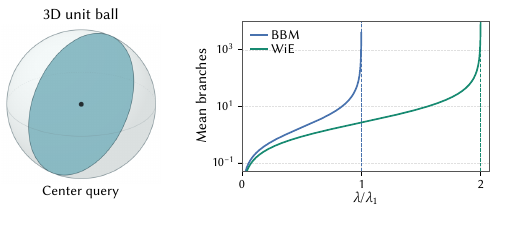}
\caption{The full-tree critical rate moves from $\lambda_1$ to $2\lambda_1$. Left: the complete unit ball, a central section, and its center query. Right: exact expected reproduction counts, with $\rho=\lambda/\lambda_1$ and $\lambda_1=\pi^2/2$. BBM gives $k/\sin k-1$, $k=\pi\sqrt\rho$; WiE gives $1.5(k_w/\sin k_w-1)$, $k_w=\pi\sqrt{\rho/2}$. Each expression is plotted only below its own critical rate.}
\Description{A complete translucent unit ball and analytical expected branching counts for BBM and constant-envelope WiE.}
\label{fig:wiefig-criticality}
\end{figure}

\paragraph{Radial validation.}
Figure~\ref{fig:wiefig-radial} compares WiE with an independent radial boundary-value solve on the three-dimensional unit ball.

\begin{figure}[tbp]
\centering
\includegraphics[width=\columnwidth]{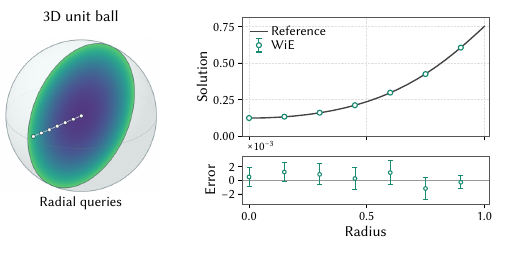}
\caption{Radial validation at $\lambda=2.1\lambda_1$ with boundary value $0.75$. The full sphere encloses a reference section and seven query locations. The upper plot compares WiE with an independent radial BVP; the lower plot resolves the signed estimation error. Points pool three seeds of 8,192 samples; bars show 95\% Monte Carlo confidence intervals. This test uses certified radial envelopes.}
\Description{A full sphere with seven radial queries beside the solution profile and signed estimation errors with confidence intervals.}
\label{fig:wiefig-radial}
\end{figure}

\paragraph{Additional three-dimensional tests.}
The following mesh tests solve $\frac12\Delta u+\lambda(u^2-u)=0$ in the volumes enclosed by Fandisk, Bunny, Spot, and Rocker-arm~\cite{jacobsonTestModels,stanfordScanning,fcpwGeometry,crane2013robust}, with $g(x)=h(x)=0.9-0.7(z_{\max}-z)/(z_{\max}-z_{\min})$, analytic affine envelopes $(\ell,a)=(0,h)$, $\lambda=2.1\widehat\lambda_1$, and boundary tolerance $10^{-6}$.
Here $\widehat\lambda_1$ is a finite-difference spectral estimate; these tests use the affine certificate without spline optimization or the preprocessing fallback.
The reference grids are $128^3$ for Fandisk, Bunny, and Spot, and $256^3$ for Rocker-arm; the last refinement changes the 48 held-out query values by RMS 0.00036, 0.00021, 0.00037, and 0.00045, respectively.
For the measured error--time comparisons, BBM samples boundary payoffs as Bernoulli outcomes and both implementations stop a root once a zero factor is known; these measurements do not evaluate full-tree work.
Figure~\ref{fig:wiefig-time-error} uses one pinned core of an Intel Xeon Silver 4310; total wall time includes mesh and BVH setup and sampling, plus analytic-envelope construction for WiE, while query generation, reference solves, rendering, and spectral analysis are excluded for both methods.

\begin{figure}[tbp]
\centering
\includegraphics[width=\columnwidth]{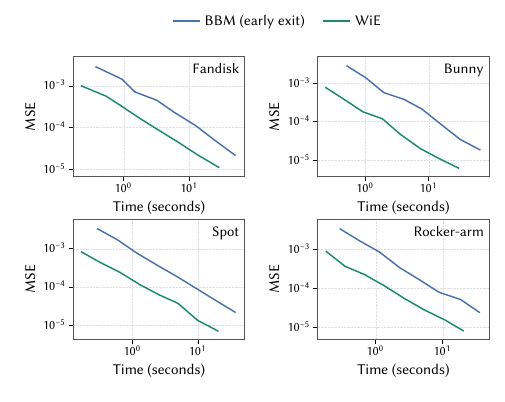}
\caption{Error versus computation time at $\lambda=2.1\widehat\lambda_1$. Each curve joins eight sample budgets $N=2^6,\ldots,2^{13}$; each point averages the MSE of three independent seeds over 48 held-out volume queries. Time includes mesh setup and WiE envelope construction. Both implementations use early exit.}
\Description{Four measured time versus mean squared error plots for Fandisk, Bunny, Spot, and Rocker-arm.}
\label{fig:wiefig-time-error}
\end{figure}

\paragraph{Envelope quality and evaluation cost.}
We further compare analytic affine envelopes with the final output of monotone B-spline construction on six meshes (Table~\ref{tab:wie-envelope-cost}), using the same PDE, height-dependent boundary data, $\lambda=2.1\widehat\lambda_1$, and boundary tolerance as above. The affine pair is $(\ell,a)=(0,h)$. Cubic tensor-product splines use $9^3$ coefficients per envelope, start from $(0,0.9)$, and allow up to eight refinement rounds, with a 60-second limit per LP. The 2,048 training points are independent of the 48 evaluation queries. Both representations use the same product estimator and exact termination after a zero factor.

At fixed $N=1{,}024$, the five valid final spline configurations reduce the measured single-sample variance by $5.4$--$21.0\%$, but require $1.63$--$4.05\times$ the sampling CPU time of the affine baseline, even before construction costs. Representing the identical affine envelopes as splines also increases sampling time by $1.17$--$1.57\times$, isolating representation and evaluation overhead. Fandisk and Spot accept no refinement and retain the initial envelope; the final Armadillo configuration requires event-probability repairs and is excluded from valid performance comparisons.

These results distinguish the variance bound from realized efficiency. The final spline envelopes have larger mean widths than the affine baseline; on Bunny, for example, the means are $0.582$ and $0.494$, respectively, despite the lower measured spline variance. Equation~\eqref{eq:var-gap} remains a valid worst-case bound, but when attainable envelopes remain broad, it can provide limited guidance for predicting actual variance or selecting the fastest representation. Lower measured variance alone is also insufficient: its benefit must offset the cost of sampling and envelope construction. For these tests, inexpensive analytic affine envelopes offer the better computational tradeoff while retaining WiE's envelope transformation.

\begin{table}[t]
\centering
\caption{Envelope comparison at a fixed sample budget. $V$ is the mean single-sample variance over 48 queries and eight independent runs; $t$ is the mean sampling CPU time in seconds for all 48 queries, with $N=1{,}024$ samples per query. Construction is excluded. A: analytic affine; S: final spline-construction output. $\dagger$: no accepted update, so S is the initial $(0,0.9)$ envelope. Dashes denote the failed Armadillo sampling audit.}
\label{tab:wie-envelope-cost}
\small
\setlength{\tabcolsep}{4pt}
\renewcommand{\arraystretch}{1.08}
\begin{tabular}{lrrrr}
\toprule
Model & $V_{\mathrm A}$ & $V_{\mathrm S}$ & $t_{\mathrm A}$ & $t_{\mathrm S}$ \\
\midrule
Octopus & 0.02514 & 0.02082 & 5.061 & 9.472 \\
Armadillo & 0.05944 & --- & 4.917 & --- \\
Fandisk$^{\dagger}$ & 0.07618 & 0.07167 & 3.183 & 10.590 \\
Bunny & 0.04761 & 0.03763 & 3.344 & 5.905 \\
Spot$^{\dagger}$ & 0.06105 & 0.05749 & 2.333 & 9.448 \\
Rocker-arm & 0.06159 & 0.05829 & 2.451 & 4.004 \\
\bottomrule
\end{tabular}
\end{table}

\paragraph{Interior field visualizations.}
Figure~\ref{fig:wiefig-spot} shows the complete Spot model and a translucent plane locating its $z=0$ query section. An independent 16-sample pilot fixes the sample count for a nominal 60-second total budget; reported times include the pilot and analytic-envelope construction. The raw field panels share query locations and color scales, without smoothing or denoising.

\begin{figure}[tbp]
\centering
\includegraphics[width=\columnwidth]{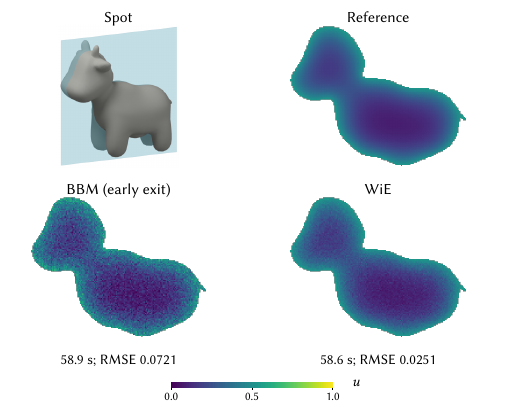}
\caption{Spot at $\lambda=2.1\widehat\lambda_1$. The full mesh and translucent blue plane locate the $z=0$ section; the flat panels show the independent $128^3$ reference and raw BBM/WiE estimates on identical queries. Times include the pilot and envelope construction; RMSE covers the entire section. The flat field panels use one unlit color scale.}
\Description{Complete Spot mesh with a translucent evaluation plane and three corresponding unlit solution fields.}
\label{fig:wiefig-spot}
\end{figure}

Figure~\ref{fig:wiefig-rocker} applies the same field-visualization protocol to Rocker-arm.

\begin{figure}[tbp]
\centering
\includegraphics[width=\columnwidth]{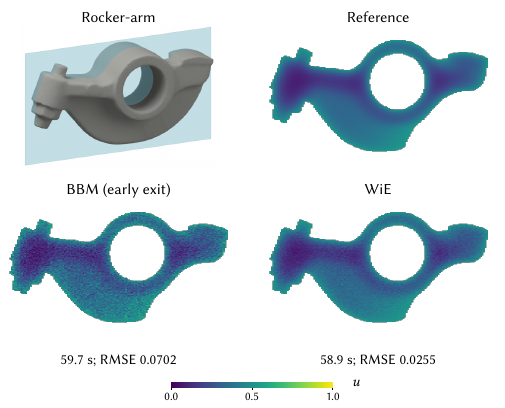}
\caption{Rocker-arm at $\lambda=2.1\widehat\lambda_1$. The full mesh and translucent blue plane locate the $z=0$ section; the flat panels show the independent $256^3$ reference and raw BBM/WiE estimates on identical queries. Times include the pilot and envelope construction; RMSE covers the entire section. The flat field panels use one unlit color scale.}
\Description{Complete Rocker-arm mesh with a translucent evaluation plane and three corresponding unlit solution fields.}
\label{fig:wiefig-rocker}
\end{figure}

\section{Discussion and Limitations}
WiE demonstrates that certified solution bounds can reshape branching dynamics rather than merely prune existing trees. While the transformation preserves unbiasedness, its efficiency depends on envelope quality and sampling acceptance rates. Constructing certified envelopes incurs preprocessing costs, making WiE less advantageous for low branching intensities or limited query workloads. Our spectral-workload fallback mitigates this overhead but relies on empirical calibration. Moreover, the near-critical complexity guarantee bounds expected reproduction events rather than total runtime. Extensions to manifold and parabolic PDEs demonstrate broader applicability, although projected walks inherit geometric approximation errors and space-time sampling introduces additional costs. Future work may include higher-dimensional certification and more efficient envelope optimization directly targeting computational work.

\bibliographystyle{ACM-Reference-Format}
\bibliography{references}
\appendix

\end{document}